\documentclass[runningheads]{llncs}
\usepackage[T1]{fontenc}
\usepackage{graphicx}
\usepackage{listings}
\usepackage{amsmath}
\usepackage{amssymb}
\usepackage{mathrsfs}
\usepackage{todonotes}

\usepackage{booktabs}
\usepackage{threeparttable}

\newcommand{\bigJobs}[1]{{J_B(#1)}}
\newcommand{\smallJobs}[1]{{J_S(#1)}}
\newcommand{\allJobs}{J}

\let\doendproof\endproof
\renewcommand\endproof{~\hfill\qed\doendproof}

\begin{document}
\title{Improved Algorithms for Monotone Moldable Job Scheduling using Compression and Convolution\thanks{Supported by DFG-Project JA 612 /25-1.}}
\titlerunning{Monotone Moldable Job Scheduling using Compression and Convolution}
% If the paper title is too long for the running head, you can set
% an abbreviated paper title here
%
\author{Kilian Grage \inst{} \and
Klaus Jansen\inst{} \and Felix Ohnesorge} 
\authorrunning{K. Grage et al.}
% First names are abbreviated in the running head.
% If there are more than two authors, 'et al.' is used.
%
\institute{Kiel University, Christian-Albrechts-Platz 4, 24118 Kiel, Germany 
\email{\{kig,kj\}@informatik.uni-kiel.de, felix-eutin@gmx.de}}
\maketitle              % typeset the header of the contribution
\begin{abstract}
In the moldable job scheduling problem one has to assign a set of $n$ jobs to $m$ machines, in order to minimize the time it takes to process all jobs. Each job is moldable, so it can be assigned not only to one but any number of the equal machines. We assume that the work of each job is monotone and that jobs can be placed non-contiguously. In this work we present a $(\frac 3 2 + \epsilon)$-approximation algorithm with a worst-case runtime of ${O(n \log^2(\frac 1 \epsilon + \frac {\log (\epsilon m)} \epsilon) +  \frac{n}{\epsilon} \log(\frac 1 \epsilon) {\log (\epsilon m)})}$ when $m\le 16n$. This is an improvement over the best known algorithm of the same quality by a factor of $\frac 1 \epsilon$ and several logarithmic dependencies. We complement this result with an improved FPTAS with running time $O(n \log^2(\frac 1 \epsilon + \frac {\log (\epsilon m)} \epsilon))$ for instances with many machines $m> 8\frac n \epsilon$. This yields a $\frac 3 2$-approximation with runtime $O(n \log^2(\log m))$ when $m>16n$.

We achieve these results through one new core observation: In an approximation setting one does not need to consider all $m$ possible allotments for each job. We will show that we can reduce the number of relevant allotments for each job from $m$ to $O(\frac 1 \epsilon + \frac {\log (\epsilon m)}{\epsilon})$. Using this observation immediately yields the improved FPTAS. For the other result we use a reduction to the knapsack problem first introduced by Mounié, Rapine and Trystram. We use the reduced number of machines to give a new elaborate rounding scheme and define a modified version of this this knapsack instance. This in turn allows for the application of a convolution based algorithm by Axiotis and Tzamos. We further back our theoretical results through a practical implementation and compare our algorithm to the previously known best result. %These results show that our algorithm runs faster in real-time as well.

\end{abstract}
\section{Introduction}
The machine scheduling problem, where one assigns jobs to machines in order to finish all jobs in a preferably short amount of time, has been a core problem of computer science. Its applications are not only limited to the usual context of executing programs on a range of processor cores but it also has many applications in the real world. For example one can view machines as workers and jobs as tasks or assignments that need to be done. It is realistic in this setting that multiple workers can work on one task together to solve it more quickly.
This however gives rise to another layer of this problem, where one has to initially assign a number of machines to each job and a starting time, leading to the problem called Parallel Task Scheduling with Moldable Jobs. Our goal is to minimize the time when the last job finishes, which is called the \textit{makespan}.\looseness=-1

In this problem the time necessary for a job to be processed is dependent on the number of assigned machines. We further consider in this paper the setting where our jobs are not only moldable but also have monotone work. The work of a job $j$ with $k$ machines is defined as $w(j,k):= t(j,k)\cdot k$, which intuitively is the area of the job. We assume that this function for a fixed job $j$ is non-decreasing in the number of machines. This assumption is natural since distributing the task on multiple machines will not reduce the amount of work but actually induce a bit of overhead due to communication among the machines.\looseness=-1

Since finding an optimal solution to this problem is NP-hard \cite{jansenLandNP13} our goal is to present approximation algorithms. Such an algorithm has to guarantee for every instance $I$ with optimal makespan $OPT(I)$ to find a solution with a makespan of at most $c \cdot OPT(I)$ for some multiplicative approximation ratio $c>1$. In this paper we introduce two algorithms that work with an accuracy $\epsilon >0$: The first guarantees an approximation ratio of $c_1 =1+\epsilon$ in time $O(n \log^2(\frac 4 \epsilon + \frac {\log (\epsilon m)} \epsilon))$ under the additional premise that $m>8 \frac n \epsilon$. Our second algorithm achieves an approximation ratio of $c_2 = \frac 3 2 + \epsilon$ with running time 
%$O(\frac n \epsilon \log^2(\frac 1 \epsilon)\log ( m + \log \frac 1 \epsilon))$ 
$O(n \log^2(\frac 1 \epsilon + \frac {\log (\epsilon m)} \epsilon) +  \frac{n}{\epsilon} \log(\frac 1 \epsilon) {\log (\epsilon m)})$ when $16n \ge m$. If we apply the first algorithm for $\epsilon = \frac 1 2$ and combine both algorithms we get an efficient $(\frac 3 2 + \epsilon)$-approximation. \looseness=-1

We achieve our results through a new core observation: Although a job can be assigned to every possible number of machines, not all $m$ different allotments may be relevant when looking for an approximate solution. In fact we will show that if $m$ is large enough we can reduce the number of relevant machine allotments to 
 $O(\frac 1 \epsilon + \frac {\log (\epsilon m)} \epsilon)$. This overall assessment is based on the concept of \textit{compression} introduced by Jansen and Land \cite{jansenLandLinTime}.\looseness=-1 
 
 We use the reduced number of relevant allotments to schedule moldable jobs via an instance of the knapsack problem. This approach was initially introduced by Mounié, Rapine and Trystram \cite{mounieTrystram}. We give a new rounding scheme to convert moldable jobs into knapsack items to define a modified version of their knapsack instance. We construct this knapsack instance in a way that the number of different sizes and profits is small. This allows for the efficient application of a knapsack algorithm introduced by Axiotis and Tzamos \cite{axiotis19} using convolution. Their algorithm works well on such instances and thanks to our rounding we can even do the required pre-processing for their algorithm efficiently in linear time.\looseness=-1

\subsection{Problem definitions and notations}

Two problems will play an important role in this paper: The first being parallel task scheduling with moldable jobs, which we will call moldable job scheduling in the following. In this problem one is given a set $\allJobs$ of $n$ jobs and a set $M$ of $m$ equal machines. We denote with $[l] = \{i\in \mathbb{N} \,|\, 1\le i \le l\}$ for any $l\in \mathbb{N}$. The processing time of a job in the moldable setting is given through a function $t:\allJobs \times [m] \rightarrow \mathbb{R}_{\ge 0}$ where  $t(j,k)$ denotes the processing time of job $j$ on $k$ machines. We denote with  $\gamma(j,d) = \min\{i\in [m] \,|\, t(j,i) \le d\}$ the number of machines required for job $j$ to achieve processing time smaller than $d$. If $d$ is not achievable with $m$ machines, we say $\gamma(j,d) $ is undefined.\looseness=-1

%Looking at one job we will often ask how many processors are required to achieve a processing time below a certain threshold $d$. To find this number of processors we denote with $\gamma(j,d) = \min(\{i\in [m] | t(j,i) \le d\} \cup \{\infty\})$ and specify with $\gamma(j,d) = \infty$ when a processing time lower or equal to $d$ is not achievable for job $j$.

For a solution of this problem we require two things: First an allotment $\alpha: \allJobs \rightarrow [m] $ and an assignment of starting times  $s: \allJobs  \rightarrow \mathbb{R}_{\ge 0}$. For simplicity we denote $\alpha_j := \alpha(j)$ and $s_j =s(j)$ respectively. A feasible solution must now fulfill that at any time at most $m$ machines are in use.  Denote with
${U(t) := \{j \in \allJobs \,|\, t \in [s_j,s_j+t(j,\alpha_j)]\} }$ the jobs that are processed at time $t$. If at all times $t\in \mathbb{R}_{\ge 0}$  we have that $\sum_{j\in U(t)}{\alpha_j} \le m$ then the schedule defined by $\alpha$ and $s$ is feasible.\looseness=-1

Finally we look to minimize the makespan of this schedule, which is the time, when the last job finishes. Given an allotment $\alpha $ and starting times $s$ the makespan is defined by 
$\max_{j\in \allJobs}\{s_j + t(j,\alpha_j)\}$. As mentioned before the work of a job is defined as $w(j,k)= k\cdot t(j,k) $. In this paper we will work under the assumption that this work function for each job is non-decreasing. More precisely for all jobs $j$ and $k,k'\in [m]$ with $k\le k' $ we have $w(j,k) \le w(j,k')$.\looseness=-1

The second main problem we will consider in this work is the knapsack problem \footnote{We mainly consider $0-1$ Knapsack, though some items may appear multiple times.}, as it will be part of our algorithm to solve a knapsack instance. In the knapsack problem one is given a set of $n$ items where each item $i$ is identified with a profit value $p_i \in \mathbb{R}_{>0}$ and a size or weight $w_i\in \mathbb{N}$. The task is to find a maximum profit subset of these items such that the total weight does not exceed a given capacity $t \in \mathbb{N}$.\looseness=-1

%Jansen and Land showed in their work that the in order to achieve a $\frac{3}{2} + \epsilon$ approximation one can reduce finding an allotment for this problem to a knapsack problem. The general idea is to guess the optimal makespan $d$ and apply a 2-shelf based algorithm that schedules jobs with processing time $\le d$ and jobs with processing time $\le \frac d 2$ in two separate shelfs. The decision where each job is allotted is decided via a constructed knapsack instance that compares the benefit in work of placing jobs in either shelf. 

%Another problem we will consider for a subroutine ($\max,+$)-convolution problem we consider in this paper one is given two sequences or arrays $(a_i)_{0\le i < n}, (b_i)_{0\le i < n}$ and has to find the convolution $c = a \oplus b$ that is defined through $c_i = \max_{j\le i}(a_j +b_{i-j})$ for all $i\in \mathbb{N}_{< n}$.

\subsection{Related work}

The moldable job scheduling problem is known to be NP-hard \cite{mallComplexity} even with monotone work functions \cite{jansenLandNP13}. Further there is no polynomial time approximation algorithm with a guarantee less than $\frac 3 2$ unless P=NP \cite{32Complexity}. Belkhale and Banerjee gave a 2-approximation for the problem with monotony \cite{Belkhale}, which was later improved to the non-monotone case by Turek et al.~\cite{Turek}. Ludwig and Tiwari improved the running time further \cite{ludwigTiwari} and achieved a running time polylogarithmic in $m$, which is especially important for compact input encoding, where the length of the input is dependent on $\log m$ and not $m$.\looseness=-1 

Mounié et al.~ gave a $(\frac 3 2 + \epsilon)$-approximate algorithm with running time $O(nm\log \frac 1 \epsilon)$ \cite{mounieTrystram}. Jansen and Land later improved this result further by giving an FPTAS for instances with many machines and complementing this with an algorithm that guarantees a ratio of $(\frac 3 2 + \epsilon)$ with polylogarithmic dependence on $m$. They picked up on the idea of Mounié et al.~ to use a knapsack instance to find a schedule distributing jobs in two shelfs and modified the knapsack problem to solve it more efficiently. In a recent result Wu et al.~\cite{WuZhangChen2022} gave a new $\frac 3 2 $- approximation that works in time $O(nm\log(nm))$\looseness=-1

The Knapsack problem as a generalization from Subset Sum is another core problem of computer science that is NP-hard as well. For this problem pseudopolynomial algorithms have been considered starting with Bellmans classical dynamic programming approach in time $O(nt)$ \cite{bellman}. Many new results with pseudopolynomial running times have recently been achieved in regards to various parameters such as largest item size or number of different items \cite{eisenbrand,polak21,axiotis19,bateni}.\looseness=-1

One interesting connection has come up between Knapsack and the $(\max,+)$-convolution problem. In this problem one is given two sequences of length $n$ $(a_i)_{0\le i < n}, (b_i)_{0\le i < n}$ and has to find the convolution $c = a \oplus b$ which is defined through $c_i = \max_{j\le i}(a_j +b_{i-j})$ for all $i\in \mathbb{N}_{< n}$. This problem can be solved in quadratic time $O(n^2)$. Cygan et al.~\cite{cygan} conjecture that a subquadratic algorithm may not be possible and used this conjecture as a basis for many fine-grained complexity results for Knapsack and similar problems. Axiotis and Tzamos showed that with concave sequences, convolutions can be computed in linear time $O(n)$ and they used this to give a $O(Dt)$ for Knapsack where $D$ is the number of different item sizes \cite{axiotis19}. This approach has also been used by Polak et al.~\cite{polak21} in conjunction with proximity arguments from Eisenbrand Weismantel \cite{eisenbrand} to gain fast algorithms for knapsack with small item sizes .\looseness=-1
%\todo{expand more, Knapsack book, Min-conv improvement}

\subsection{Our results}

We present a new algorithm, in particular a $(\frac 3 2+\epsilon)$-approximation algorithm, for any accuracy parameter $\epsilon>0$, with a runtime polynomial in $n, \frac 1 \epsilon$ and in $\log m$. Since we are polynomial in $\log m$ our algorithm will be able to handle certain compact input encodings and will generally scale well into large $m$.

The main difficulty in moldable job scheduling is that for every job we need to choose between $m$ different allotments and then schedule jobs efficiently. We will however show that not all $m$ possible allotments have to be regarded. Since we look for an approximate solution and we have monotone jobs, it is sufficient to only consider $O(\frac{1}{\epsilon} + \frac{\log(\epsilon m)}{\epsilon}))$ different machine counts. This leads immediately to a fully polynomial time approximation scheme (FPTAS) for instances with many machines.\looseness=-1

% By using compression and allowing a small $\epsilon$ factor error in the makespan, we will limit the number of allotments we consider. Instead of allowing any number of processors from $1$ to $m$ for jobs, we will specify a set of $O(\frac{1}{\epsilon} + \frac{\log(\epsilon m)}{\epsilon}))$ different processor counts and only limit ourselves to counts of this set or slightly modified counts. We will use this idea to improve the running time of the FPTAS \cite{jansenLandLinTime} on instances with many machines.

\begin{theorem}\label{Theo:FPTAS}
Let $\epsilon >0$. For moldable job scheduling with instances where ${m > 8 \frac n \epsilon}$ exists a $(1+\epsilon)$-approximation that runs in  time $O(n\log^2(\frac{1}{\epsilon} + \frac{\log(\epsilon m)}{\epsilon}))$.
\end{theorem}

This result can be used for a $\frac 3 2 $-approximation if we use $\epsilon= \frac 1 2$. 

\begin{corollary}
Consider moldable job scheduling on instances with $m >16 n$. There exists a $\frac 3 2$-approximation in time $O(n\log^2(\log m))$.
\end{corollary}

We complement this result with an efficient $(\frac 3 2+\epsilon)$-approximation for the case where $m \le 16 n$. For this we follow the same approach as \cite{jansenLandLinTime,mounieTrystram} and we aim to construct a knapsack instance. We will introduce a new rounding scheme for machine counts, processing times and job works and convert these modified jobs into knapsack items. The resulting knapsack instance will only have a small amount of different item sizes. We then apply an algorithm introduced by Axiotis and Tzamos \cite{axiotis19} that works well on such instances. Thanks to our rounding we will be able to do the pre-processing of their algorithm in linear time as well.

%We further want to give an efficient $\frac 3 2$-approximation for the case where $m \le 16 n$. To achieve this we use the same approach by constructing a knapsack instance. We however introduce a new and more simple rounding scheme for processor times and work, while using the same processor counts from the FPTAS. This yields a knapsack instance with not many different item sizes and profits. We therefore propose to use an approach introduced by Axiotis and Tzamos \cite{axiotis19} to solve the knapsack problem efficiently using convolution. 
%Since our rounding scheme gives us an easy structure for the profits we are able to show that the required preprocessing for Axiotis' algorithm can be done quite quickly.\looseness=-1

\setlength{\tabcolsep}{4pt}
\begin{theorem}\label{Theo:ThreeHalf}
For moldable job scheduling there exists an algorithm that for instances with $m\le 16n$ and for any $\epsilon >0$ yields a $\frac 3 2 +\epsilon$ approximation in time: 
$O(n \log^2(\frac 1 \epsilon + \frac {\log (\epsilon m)} \epsilon) +  \frac{n}{\epsilon} \log(\frac 1 \epsilon) {\log (\epsilon m)})$

%\begin{itemize}
%\item $O(n \log^2(\frac 1 \epsilon + \frac {\log (\epsilon m)} \epsilon) +  \frac{m}{\epsilon} \log(\frac 1 \epsilon) {\log (\epsilon m)})$ when $m \le 16n$.
%\item $O(n \log^2(m) +  \frac{m}{\epsilon} \log(\frac 1 \epsilon))$ when $ m\le \frac 4 \epsilon$.
%\item $O(n \log^2(m) +  \frac{n}{\epsilon} \log(\frac 1 \epsilon) )$ when $m\le \frac 4 \epsilon$.
%\end{itemize}
\end{theorem}

These two results make up one $(\frac 3 2 + \epsilon)$-approximation that improves on the best known result by Jansen and Land \cite{jansenLandLinTime} in multiple ways. For large $m$ we manage to reduce the dependency on $m$ even further. When $m$ is small we improve on their running time by reducing the dependency on $\epsilon$ by a factor of $\frac 1 \epsilon$ and several polylogarithmic factors. We also argue that our algorithm is overall simpler compared to theirs, as we do not require to solve knapsack with compressible items in a complicated manner. Instead our algorithm merely constructs the modified knapsack instance and delegates to a simple and elegant algorithm from Axiotis and Tzamos \cite{axiotis19}.

%We also note that in the case  where $m \le \frac 4 \epsilon$ is our algorithm simplifies significantly as some parts of the introduced rounding do not have to be applied.

%The running time for large $m$ may seem problematic in terms of compact input encoding, but we may observe that we may limit ourselves to the case where $m\le 16n$. If $m> 16 n$ we can simply apply theorem \ref{Theo:FPTAS} with $\epsilon= \frac 1 2$ to get a solution in time $O(n\log^2(\log(m)))$.
%Our running times are summarized and compared to the known algorithm from Jansen and Land in the following table. As we can see we improve the runtime in multiple cases, but especially in the worst case by a $\frac 1 \epsilon$ factor and several polylogarithmic factors.
\begin{center}
	\small
	\begin{threeparttable} 
		\begin{tabular}{lll}
			\toprule
			\textsc{Result} & \textsc{Jansen \& Land \cite{jansenLandLinTime}} & \textsc{This paper} \\
			\midrule
			\textsc{$1+\epsilon$, $(m>8\frac n \epsilon)$ } 
			& $O(n \log(m)(\log(m) + \log (\frac 1 \epsilon)))$ 
			& $O(n \log^2(\frac 1 \epsilon + \frac {\log (\epsilon m)} \epsilon))$ \\
			\midrule
			\textsc{$\frac 3 2$, $(m>16n)$}
			& $O(n \log^2(m))$ 
			& $O(n \log^2(\log m))$ \\
			\midrule
			\textsc{\textsc{$\frac 3 2 +\epsilon$, $(m\le 16n)$}}  
			& $O(\frac n {\epsilon^2} \log m (\frac{\log m}{\epsilon}+\log^3(\epsilon m)))$ 
			& $O(n \log^2(\frac 1 \epsilon + \frac {\log (\epsilon m)} \epsilon) +  \frac{n}{\epsilon} \log(\frac 1 \epsilon) {\log (\epsilon m)})$\\
            \bottomrule
			%\midrule
			%\textsc{$\frac 3 2 +\epsilon$, $(m\le \frac 4 \epsilon)$} 
			%& $O(\frac n {\epsilon^2} \log m (\frac{\log m}{\epsilon}+\log^3(\epsilon m)))$ %\tnote{2}
			%& $O(n \log^2(m) +  \frac{m}{\epsilon} \log(\frac 1 \epsilon))$ \\			
		\end{tabular}		
	\end{threeparttable} 
\end{center}

\section{General Techniques and FPTAS for many machines}

%\subsection{Compression and Dual Approximation Framework}\label{Sec:DualFW}

The core technique used in this paper is the concept of \textit{compression} introduced by Jansen and Land \cite{jansenLandLinTime}. Compression is the general idea of reducing the number of machines a job is assigned to. Due to monotony the resulting increase of processing time can be bound. \looseness=-1

\begin{lemma}[\cite{jansenLandLinTime}]\label{Lem:Comp}
Let $\rho \in (0,1/4]$ be what we denote in the following as a \textit{compression factor}. Consider now a job $j$ and a number of machines $k\in \mathbb{N}$ with $\frac{1}{\rho} \le k \le m$, then we have that $t(j, \lfloor (1-\rho)k \rfloor) \le (1+4\rho) t(j,k)$.
\end{lemma}

The intuitive interpretation of this lemma is that if a job uses  $k \ge \frac{1}{\rho}$ machines then we can free up to $\lceil \rho k \rceil$ machines and the processing time increases by a factor of $4\rho$. We are going to use this lemma in the following by introducing a set of predetermined machine counts.\looseness=-1

\begin{definition}\label{Def:Sizes}
Let $\rho$ be a compression factor and set $b := \frac 1 \rho$. We define 
$S_\rho := [b] \cup \{ \lfloor (1+\rho)^i b \rfloor \,|\, i\in [\, \lceil \log_{1+\rho}(\frac m b) \rceil \,] \}$ as the set of $\rho$-compressed sizes.
\end{definition}
 Note that reducing machine numbers to the next smaller size in $S_\rho$ corresponds to a compression and processing time may only increase by a factor of at most $1+4\rho$.\looseness=-1
 %and thus increases the processing time of any job only by at most .\looseness=-1

\begin{corollary}\label{Cor:Sizes}
Let $\epsilon \in (0,1)$ be an accuracy parameter then $\rho = \frac \epsilon 4$ is a compression factor and $|S_\rho | \in O(\frac 1 \epsilon + \frac {\log(\epsilon m)} \epsilon)$.
\end{corollary}

Generally our algorithms will work on the set $S_\rho$ for $\rho = \frac 4 \epsilon$ and only assign machine counts in $S_\rho$. If $m \le \frac 4 \epsilon$ we  work with any machine number as $S_\rho = [m]$. The algorithms we present will work in a dual approximation framework.\looseness=-1

%We will generally assume in the following that $m \ge \frac 4 \epsilon$ as we otherwise work with a constant number of possible assignments for each job to begin with. We then use these $\rho$-compressed processor sizes to reduce the complexity of our instance. This idea will be incorporated in the dual approximation framework and also the later $\frac 3 2$-approximation. 

A dual approximation framework is a classical approach for scheduling problems. The general idea is to use an approximation algorithm with constant ratio $c$ on a given instance and gain a solution with makespan $T$. While this is only an approximation we can conclude that the makespan $T^*$ of an optimal solution must be in the interval $[\frac{T}{c},T]$ and we can search this space via binary search. We can then see a candidate $d\in [\frac{T}{c},T]$ as a guess for the optimal makespan.\looseness=-1 

The approximation algorithm is then complemented with an estimation algorithm, that receives an instance $I$ and a guess for the makespan $d$ as input. This estimation algorithm then must be able to find a schedule with a makespan of at most $(1+\epsilon)d$ if such a schedule exists. If $d$ was chosen too small, i.e. $(1+\epsilon)d < OPT(I)$, our algorithm can reject the value $d$ and return false. 

We continue to apply this algorithm for candidates, until we find $d$ such that the algorithm is successful for $d$ but not for $\frac{d}{1+\epsilon}$. Note that if the algorithm fails for $\frac{d}{1+\epsilon}$ we have that $d = (1+\epsilon) \frac{d}{1+\epsilon}< OPT(I)$. Therefore the solution generated for $d$ has a makespan of $(1+ \epsilon)d < (1+ \epsilon)OPT(I)$. Using binary search we can find such a candidate $d$ in $O(\log \frac{1}{\epsilon})$ iterations \cite{jansenLandLinTime}.\looseness=-1

\subsection{Constant factor approximation}

Our constant factor approximation is gonna work in two steps:
First we compute an allotment and assign each job to a number of machines. Secondly we will use list scheduling in order to schedule our now fixed parallel jobs.For the first step we use an algorithm introduced by Ludwig and Tiwari \cite{ludwigTiwari}.\looseness=-1

%\begin{lemma}\label{Lem:Allot}
%Given an instance $J$ for Moldable Parallel Task Scheduling with $n$ jobs and $m$ machines we can compute an allotment $\alpha:[n]\rightarrow[m]$ such that the value 
%\begin{equation*}
%\omega = \min(\frac{1}{m} \sum\limits_{j\in J}w(j,\alpha(j)), \max_{j\in J}t(j,\alpha(j)))
%\end{equation*}
%is minimized among all possible allotments in time $O(n \log^2 m)$.
%
%Furthermore if we only look for an allotment $\alpha:[n]\rightarrow S$ for some $S\subseteq [m]$, we can compute such an allotment in time $O(n \log^2 |S|)$.
%\end{lemma}

\begin{lemma}[\cite{ludwigTiwari}]\label{Lem:Allot}
Let there be an instance $I$ for moldable job scheduling with $n$ jobs and $m$ machines. For an allotment $\alpha:J\rightarrow[m]$  we denote with 
\begin{equation*}
 \omega_{\alpha} := \min(\frac{1}{m} \sum\limits_{j\in J}w(j,\alpha(j)), \max_{j\in J}t(j,\alpha(j)))
\end{equation*}
 the trivial lower bound for any schedule that follows the allotment $\alpha$.
Furthermore for $S\subseteq [m]$ we denote with $\omega_{S} := \min\limits_{\alpha:J\rightarrow S} \omega_{\alpha}$ the trivial lower bound possible for any allotment, which allots any job to a number of machines in $S$.

 For any $S\subseteq [m]$  we can compute an allotment $\alpha:J\rightarrow S$ with $\omega_{\alpha} = \omega_{S}$ in time $O(n \log^2 |S|)$.
\end{lemma}

%Jansen and Land used this result in a trivial way to find an allotment $\alpha:[n]\rightarrow[m]$ using arbitrary machine numbers.
%A very common approach to moldable job scheduling is to first compute an allotment $\alpha:[n]\rightarrow[m]$ and then use a parallel task scheduling algorithm to schedule the jobs.  We intend to follow this approach, but instead of computing an arbitrary allotment we will specify a reduced set of processor sizes, by rounding the processor sizes geometrically and use this lemma to compute an allotment faster. By doing so we will not find the globally minimal lower bound $\omega_{[m]}$, but we will see that we can approximate the best lower bound and therefore give an approximation algorithm with a slightly worse multiplicative approximation ratio.
%Our goal is to reduce the runtime by using the second part of this lemma and reducing the number of relevant processor sizes. We cannot guarantee finding the minimal value for $\omega$ but as we will conclude in the following we can find an approximate value $\omega' \le (1+\epsilon)\omega$ by looking for an allotment in a reduced set of processor sizes.

%We will use a rounding procedure to reduce the number of relevant processor counts. (TODO Specify: Either reduce machine count itself to relevant machine count number or round up heights of non-relevant width to height of next smallest relevant width.) Let $D$ denote the set of relevant processor counts and note that $|D|\in O(\frac{\log(\delta m)}{\delta})$.

We apply this lemma but limit machine numbers to $\rho$-compressed sizes $S_\rho $ for $\rho = \frac \epsilon 4$. With that we gain an approximate value of $\omega_{[m]}$
\looseness=-1

%We will use this algorithm to compute an appropriate allotment. We will not consider an arbitrary allotment but limit processor numbers to $\rho$-compressed item sizes $S_\rho $ for $\rho = \frac \epsilon 4$. With that we gain an approximate value of $\omega_{[m]}$.\looseness=-1

\begin{lemma}\label{Lem:Allot2}
Given an instance $I$ for moldable job scheduling with $n$ jobs, $m$ machines and accuracy $\epsilon <1$. In time $O(n \log^2 (\frac{4}{\epsilon} +\frac{\log(\epsilon m)}{\epsilon}))$ 
we can compute an allotment $\alpha:J\rightarrow [m]$ such that $\omega_{\alpha} \le (1+\epsilon) \omega_{[m]}$.
%\begin{equation*}

%\end{equation*}

\end{lemma}

\begin{proof}

Let $\rho = \frac \epsilon 4, b = \frac 1 \rho$ and $S_\rho$ be the set of $\rho$-compressed sizes by definition \ref{Def:Sizes}. We now use lemma \ref{Lem:Allot} to compute an allotment $\alpha':[n]\rightarrow S_\rho$ such that $\omega_{\alpha'} = \omega_{S_\rho}$ and note that the proposed running time follows from corollary \ref{Cor:Sizes} and lemma \ref{Lem:Allot}. It remains to show that $\omega_{\alpha'} \le (1+\epsilon) \omega_{[m]} $.

%For that set $\epsilon' := \epsilon/4$ and note that $\epsilon'$ is a valid compression factor. Set $b := \lceil 4 / \epsilon \rceil $ and let $S:= [b] \cup \{b (1+\epsilon')^i | i \in [\log_{1+\epsilon'}(\frac{m}{b})] \}$. We now use lemma \ref{Lem:Allot} to compute an allotment $\alpha':[n]\rightarrow S$ such that $\omega_{\alpha'} = \omega_{S}$. Note that $|S| = b + \log_{1+\epsilon'}(\frac{m}{b}) $ leading to the proposed runtime. 

For this let $\alpha$ be an allotment with $\omega_{\alpha}= \omega_{[m]}$. We now modify this allotment by rounding its assigned number of machines down to the next value in $S_\rho$. To be more precise let  $\alpha'':[n]\rightarrow S_\rho;j \mapsto \max\{s\in S_\rho |s \le \alpha(j)\}$. Note that based on the definitions and lemma \ref{Lem:Allot} it follows immediately that 
$\omega_{\alpha} \le \omega_{\alpha'} \le \omega_{\alpha''}$. We will conclude the proof by showing that $\omega_{\alpha''} \le (1+\epsilon)\omega_{\alpha}$.

We note that the rounding from $\alpha$ to $\alpha''$ is a compression. To see that consider two consecutive item sizes $\lfloor b(1+\rho)^{(i-1)}\rfloor,\lfloor b(1+\rho)^{(i)}\rfloor$ for some $i$ and note that: 

\begin{align*}
\lfloor b(1+\rho)^{(i)} \rfloor - \lfloor b(1+\rho)^{(i-1)} \rfloor  &\le
    b(1+\rho)^{(i)}  - ( b(1+\rho)^{(i-1)}  -1)\\
&=   b(1+\rho)^{(i)}  -  b(1+\rho)^{(i-1)}  +1\\
&= \rho b(1+\rho)^{(i-1)}  +1 \le \rho b(1+\rho)^{(i)} 
\end{align*}

Since we only round a job down when $\alpha(j) < \lfloor b(1+\rho)^{(i)}\rfloor$ we get that $\alpha(j)-\alpha''(j) \le \rho \alpha(j)$. According to lemma \ref{Lem:Comp} the  processing time of the job may only increase by a factor of at most $1+4\rho = 1+ \epsilon$. Therefore we have 
\begin{equation*}
\max_{j\in J}t(j,\alpha''(j)) \le \max_{j\in J}\{(1+\epsilon)t(j,\alpha(j))\} = (1+\epsilon) \max_{j\in J}t(j,\alpha(j)).
\end{equation*}
 Since the work function is monotone $\omega_{\alpha''} \le (1+\epsilon)\omega_{\alpha}$ follows directly.

%By applying the proposed rounding procedure any job loses at most a $\rho$ fraction of the machines it had assigned previously in $\alpha$. To see that assume that $\alpha(j)\in (\lfloor b(1+\rho)^{(i-1)}\rfloor,\lfloor b(1+\rho)^{(i)}\rfloor)$ for some $i\ge 1$ then we have that 
%\begin{align*}
%\alpha(j) - \alpha''(j) &<  \lfloor b(1+\rho)^{(i)} \rfloor - (\lfloor b(1+\rho)^{(i-1)} \rfloor +1)\\
%&\le  (1+\rho)b(1+\rho)^{(i-1)} - (b(1+\rho)^{(i-1)} -1 +1)\\
%&= \rho b(1+\rho)^{(i-1)} < \rho \alpha(j) \le \lceil \rho \alpha(j) \rceil .
%\end{align*}

%\\ $\frac{1}{m} \sum\limits_{j\in J}w(j,\alpha''(j)) \le \frac{1}{m} \sum\limits_{j\in J}w(j,\alpha(j))$. If we put these two conclusion together we finally get that $\omega'' \le (1+\epsilon)\omega$.
\end{proof}

With this allotment we use list scheduling to achieve a constant factor approximation \cite{gareyGraham}. We use this in our dual-approximation framework. In thenext sections we will assume that we are given a makespan guess $d$ and give the required estimation algorithms for the desired results.\looseness=-1

\begin{corollary}\label{Cor:Guess}
The proposed algorithm is an approximation algorithm with a multiplicative ratio of $4$ and requires time $O(n \log^2 (\frac{4}{\epsilon} +\frac{\log(\epsilon m)}{\epsilon}))$.
\end{corollary}
%\todo{Brauch man hier noch n log n für List-Scheduling? Glaube eigentlich nicht, check ich noch im detail}

\begin{proof}
%Let $\alpha'$ the allotment based on lemma \ref{Lem:Allot2} and $\omega_{\alpha'}$ its respective lowerbound. Let $OPT(I_{\alpha'})$ be the minimal makespan of any schedule that uses allotment $\alpha'$, then we know that $OPT(I_{\alpha'}) \le (1+\epsilon)OPT(I)$, since we can take an optimal solution for $I$ and round the processing times down to the next value in $S_{\frac \epsilon 4}$. That will increase the the makespan of the optimal solution by a factor $\epsilon$ and this modified solution will 

%Then list-scheduling will yield a schedule with makespan $M\le 2 \omega' \le 2 (1+\epsilon) OPT(I)$.

The running time results mainly from applying lemma \ref{Lem:Allot2} to gain an allotment $\alpha$ with $\omega_{\alpha}\le (1+\epsilon) \omega_{[m]}$. Applying list scheduling to our computed allotment yields a schedule with makespan $2\omega_{\alpha}\le 2(1+\epsilon) \omega_{[m]} \le 4 OPT(I)$.
\end{proof}

%With this approximation algorithm we can use our dual-approximation framework. In the following sections we will therefore assume that we are given a makespan guess $d$ and give the required estimation algorithms for the desired results.\looseness=-1

\section{FPTAS for large machine counts}

In the following we assume that for every instance we have $m > 8\frac{n}{\epsilon}$. Jansen and Land showed that an FPTAS can be achieved by simply scheduling all jobs $j$ with $\gamma(j,(1+\epsilon)d)$ machines at time $0$. They consider all possible number of machines for each job. We argue that it is sufficient to consider assigning a number in $S_{\frac{\epsilon}{4}}$ to achieve a similar result. We will however require another compression to make sure our solution is feasible. \looseness=-1
%Jansen and Land showed that compression enables a very fast and efficient algorithm via a dual approximation framework. The idea behind their algorithm is quite easy as one takes a makespan guess $d$ and simply allots the minimum number of machines to a job such that each job has a processing time of less then $(1+\epsilon)d$. Due to the large number of machines and via compression Jansen and Land then argue, that this allotment assigns a total number of $m$ machines to all jobs, which yields a trivial schedule where all jobs are scheduled next to each other.

%We want to show that this procedure can be improved slightly. We've already proven previously that finding the lower and upper bound for guessing the makespan can be done faster using rounded processing times. The main running time of Jansens and Lands algorithm lies in the fact that they iterate overall $m$ possible processor counts via binary search to find the desired number of processors. We show in the following that it is sufficient to only consider rounded processor numbers.

\begin{lemma}\label{Lem:FPTASBS}
Given an instance $I$ with $n$ jobs, $m > 8\frac{n}{\epsilon}$ machines and a target makespan $d$, we can in time $O(n\log(\frac{4}{\epsilon} + \frac{\log(\epsilon m)}{\epsilon}))$ find a schedule with makespan $(1+3\epsilon)d$ if $d\ge OPT(I)$ or confirm that $d < OPT(I)$.\looseness=-1
\end{lemma}

\begin{proof}
Let $S_\rho$ be the set of $\rho$-compressed sizes for $\rho = \frac \epsilon 4$ and $b= \frac 1 \rho$.
Let $\gamma'(j,d) := \max\{s\in S_\rho | s\le  \gamma(j,d) \}$ and denote a job as \textit{narrow} when $\gamma'(j,d) \le b$ or \textit{wide} when $\gamma'(j,d) > b$. The schedule we propose results from scheduling narrow jobs with $\gamma'(j,d)$ machines and wide jobs with a compressed number of machines, that is $\lfloor (1-\rho)\gamma'(j,d)\rfloor$. We schedule all jobs at time $0$ next to each other. The running time results from 
finding $\gamma'(j,d)$ for all jobs via binary search. Note that if $\gamma'(j,d)$ is undefined for some job, then $d$ was chosen too small.\looseness=-1

Every job $j$ scheduled with $\gamma(j,d)$ machines has processing time of at most $d$. Rounding down the number of machines to $\gamma'(j,d)$ may increase the processing time by a factor of $1+4\rho$, as this process corresponds to a compression. We then apply another compression to wide jobs, which may increase the processing time again by the same factor. In total the new processing time of a job is bound by : $(1+4\rho)((1+4\rho)t(j,\gamma(j,d))) \le (1+\epsilon)^2 d\le (1+3\epsilon)d$.\looseness=-1

It remains to show that our schedule uses at most $m$ machines in total. Jansen and Land showed that $\sum_{j\in J}\gamma(j,d) \le m + n$. We assume that $\sum_{j\in J}\gamma(j,d) > m$, since otherwise our schedule would be feasible already. Denote with $J_W,J_N$ the set of wide and narrow jobs. We can see that that  $\sum_{j\in J_N}\gamma(j,d) \le  n \cdot b = 4 \frac n \epsilon < \frac 1 2 m$ and therefore $\sum_{j\in J_W}\gamma(j,d)> \frac 1 2 m $. We will show that our rounding and compression procedure will free up enough machines.

Consider a wide job $j$ and write $\gamma(j,d) = \gamma'(j,d) +r $ for some $r$. Since $j$ was assigned to $\lfloor (1-\rho)\gamma'(j,d)\rfloor$ machines, the number of freed up machines is at least:

\begin{align*}
\gamma(j,d) - \lfloor (1-\rho)\gamma'(j,d)\rfloor &\ge \gamma'(j,d) +r - (1-\rho)\gamma'(j,d) \\
&= \rho\gamma'(j,d) +r \\
&\ge \rho(\gamma'(j,d) +r) =
\rho(\gamma(j,d))
\end{align*}

In total we free at least  $\sum_{j\in J_W}(\rho\gamma(j,d)) > \rho \frac 1 2 m > \frac \epsilon 4 4 \frac n \epsilon = n $ machines. Our schedule therefore uses at most  $\sum_{j\in J}\gamma(j,d) - n \le m + n - n = m$ machines.
\end{proof}

%\begin{lstlisting}[caption={FPTAS for many machines},%backgroundcolor = \color{lightgray},
%label=list:8-6,captionpos=t,abovecaptionskip=-\medskipamount, escapeinside={(*}{*)},basicstyle=\ttfamily\linespread{1.15}\footnotesize]
%For (*$j \in J $*) do%
%	let (*$ \alpha_j = \max\{s\in S\ | s\le \gamma(j,d)\}$*)
%	if (*$ \alpha_j > b$*) then
%		 (*$ \alpha_j =  \frac{1}{1+\epsilon'}\alpha_j $*)
%if (*$ \sum_{j \in J} {\alpha_j} > m $*)
%	return false
%Schedule all jobs in parallel at time 0
%\end{lstlisting}

Note that we can apply this lemma for $\epsilon' = \frac \epsilon 3$ or an even more simplified algorithm thats results by rounding down $\gamma(j,(1+\epsilon)d)$, which also allows a simple schedule with less than $m$ machines \cite{jansenLandLinTime}. If we use this algorithm in our dual approximation framework we achieve the desired FPTAS.\looseness=-1

%\begin{corollary}
%For machine scheduling with monotone moldable jobs and any instance with $m > 8 n/\epsilon$ machines there exists an algorithm that for any $\epsilon >0$ yields a $1+\epsilon$ approximation in time $O(n\log(\frac{4}{\epsilon} + \frac{\log(\epsilon m)}{\epsilon}))$
%\end{corollary}

\begin{proof}[of Theorem \ref{Theo:FPTAS}]
We conclude for the runtime that we have to apply our dual approximation framework, meaning we apply the constant factor approximation  and then for $\log(\frac 1 \epsilon)$ makespan guesses we apply lemma \ref{Lem:FPTASBS}. Combining these running times we get a time of 
%$O(n \log^2 (\frac{4}{\epsilon} +\frac{\log(\epsilon m)}{\epsilon}) + \log(\frac 1\epsilon) n\log(\frac{4}{\epsilon} + \frac{\log(\epsilon m)}{\epsilon})) \subseteq 
$O(n \log^2 (\frac{1}{\epsilon} +\frac{\log(\epsilon m)}{\epsilon}))$.\looseness=-1
\end{proof}

\section{$(\frac{3}{2}+ \epsilon)$-Approximation}

We will now consider the goal of achieving a $\frac{3}{2}+ \epsilon$ multiplicative approximation ratio. Our algorithm will operate again in the context of the dual approximation framework. Therefore we assume a makespan guess $d$ and give an estimation algorithm. Our estimation algorithm will reduce the scheduling problem to a knapsack instance in a way that was initially introduced by Mounié et al.~\cite{mounieTrystram}. This approach was also used by Jansen and Land \cite{jansenLandLinTime} who gave a modified version of this knapsack instance. We however propose a new simpler rounding scheme that uses $\rho$-compressed sizes for $\rho = \frac 4 \epsilon$ and further modify item profit. In that way we do not need a complicated algorithm to solve the knapsack problem, but we can actually apply the result from Axiotis and Tzamos \cite{axiotis19} in an efficient manner. \looseness=-1

%Jansen and Land \cite{jansenLandLinTime} then proposed a heavily modified version of this knapsack problem, rounding processor counts and work values, and introduced an algorithm catered to their instance. We propose a new and differently modified version of this knapsack problem. Our main goal is to keep the number of different item sizes and profits in the knapsack problem small. We therefore use again $\rho$-compressed item sizes for $\rho = \frac 4 \epsilon$ and specifically round work values or processing times depending on whether jobs use many machines in either shelf. We keep the number of item profits and item sizes small in order to apply an algorithm from Axiotis and Tzamos that works efficiently on such knapsack instances \cite{axiotis19}.\looseness=-1

% Again the algorithm in this section will follow the same main idea of Jansen and Land \cite{jansenLandLinTime} who mainly used a reduction to knapsack based on Mounié et, al. \cite{mounieTrystram}. Their main idea was to distribute all jobs between two shelfs. In the first shelf $S_1$ all assigned jobs $j$ would be allotted to $\gamma(j,d)$ processors while in the second shelf $S_2$ every job $j$ would be assigned to $\gamma(j,d/2)$ processors. This assignment of jobs may not be enough to schedule all jobs in these two shelfs, but Mounié et al. introduced a few transformation rules that transforms an invalid assignment of jobs into a feasible schedule.\looseness=-1

At the start we split the set of jobs in small and big jobs $\allJobs = \bigJobs{d} \cup \smallJobs{d}$ with $\smallJobs{d} := \{j\in \allJobs \,|\, t(j,1) \le  \frac d 2\}$ and $\bigJobs d = \allJobs \backslash \smallJobs d$. 
Since we can add small items greedily at the end in linear time \cite{jansenLandLinTime}, we only need to schedule large jobs. We give a short run-down on the most important results in regards to the knapsack instance introduced by Mounié et al. .
\looseness=-1

Their main idea was to distribute all jobs into two shelfs with width $m$. The first shelf $S_1$ has height $d$ and the second shelf $S_2$ has height $\frac d 2$. If a job $j$ was scheduled in either shelf with height $s\in \{d, \frac d 2\}$ then $j$ would be allotted to $\gamma(j,s)$ machines. In order to assign jobs to a shelf, they use the following knapsack instance:\looseness=-1

Consider for each job $j \in \bigJobs d$ an item with size $s_j(d) := \gamma(j,d)$ and profit $p_j(d) := w(j,\gamma(j,d/2))- w(j,\gamma(j,d))$ and set the knapsack size to $t := m$. Intuitively this knapsack instance chooses a set of jobs $J'$ to be scheduled in $S_1$. These jobs are chosen such that their work increase in the $S_2$ would be large. 

We will denote this problem as $KP(\bigJobs d,m,d)$ where the first two parameters declare the items and knapsack size and the third parameter is the target makespan, which then determines the size and profits of the items. Given a solution $J'\subseteq J$ we denote the total work of the resulting two-shelf schedule by $W(J',d)$ and note that: \looseness=-1

\begin{align*}
W(J',d) &= \sum_{j \in J'}{w(j,\gamma(j,d))} + \sum_{j\in \bigJobs{d} \backslash J'} {w(j, \gamma(j, \frac d 2 ))} \\
&= \sum_{j\in \bigJobs d}{w(j,\gamma(j,\frac d 2))} + \sum_{j\in J'} w(j, \gamma(j,  d )) - \sum_{j\in J'} w(j, \gamma(j, \frac d 2 )) \\
&= \sum_{j\in \bigJobs d}{w(j,\gamma(j,\frac d 2))} - \sum_{j\in J'} p_j(d)
\end{align*}

As the knapsack profit is maximized, the total work $W(J',d)$ is minimized.
The result from Mounié et al. which we use is summarized in these two lemmas. We refer to either \cite{jansenLandLinTime,mounieTrystram} for a detailed description of these results.\looseness=-1

\begin{lemma}[\cite{mounieTrystram}]\label{Lem:schedtoks}
If there is a schedule for makespan $d$, then there is a solution $J' \subseteq \bigJobs d$ to the knapsack instance with $W(J',d) \le md - W(\smallJobs d,d)$.\looseness=-1
\end{lemma}

\begin{lemma}[\cite{mounieTrystram}]\label{Lem:kstosched}
If there is a solution $J' \subseteq \bigJobs d$ to the knapsack instance with $W(J',d) \le md - W(\smallJobs d,d)$, then  we can find a schedule for all jobs $J$ with makespan $\frac{3}{2}d$ in time $O(n \log n)$.\looseness=-1
\end{lemma}

Based on these lemmas we can easily reject a makespan guess $d$ if $W(J',d)$ is larger than $md-W(\smallJobs d,d)$. We note as well that lemma \ref{Lem:kstosched} can be applied if we find a solution for a higher makespan.\looseness=-1

\begin{corollary}[\cite{jansenLandLinTime}]\label{coro:final}
Let $d'\ge d$ and $J' \subseteq \bigJobs d$ be a feasible solution of the knapsack problem $KP(\bigJobs d,m, d')$ with $W(J',d) \le md'-W(\smallJobs d,d)$. Then we can find a schedule with makespan at most $\frac 3 2 d'$ in time $O(n\log n)$.\looseness=-1
\end{corollary}

We now construct a modified knapsack instance in order to apply this corollary for $d' = (1+4\epsilon)d$. First of all we reduce machine counts to $\rho$-compressed sizes for $\rho = \frac \epsilon 4$ . Consider $S_\rho$ and $b := \frac 1 \rho$ and let 
${\gamma'(j,s):= \max\{k \in S_\rho | k \le \gamma(j,s)\} }$
 for any job $j$ and $s \in \{\frac d 2, d\}$. With 
${\widetilde{p_j}(d) := \gamma'(j,\frac d 2)t(j,\gamma'(j,\frac d 2) - \gamma'(j,d)t(j,\gamma'(j,d)}$ denote the intermediary profit that is going to be further modified.

We further consider a job \textit{wide} in a shelf if it uses more than $b$ machines in the respective shelf, that is if $\gamma'(j,s)\ge b$ for the respective $s \in \{\frac d 2, d\}$. If a job is not wide we call it \textit{narrow} instead, with respect to some shelf. \looseness=-1

\begin{sloppypar}
For jobs that are narrow in both shelfs we will directly modify the profits. Let $j$ be a job with $\gamma'(j,s)< b$ for both $s\in \{\frac d 2, d\}$, then we round the intermediary profit up to the next multiple of $\epsilon d$ by setting 
${ p'_j(d) := \min\{i \epsilon d \,|\, i \epsilon d \ge\widetilde{p_j}(d)\text{ and }  i\in \mathbb{N}^*_{\le  \frac 2 {\epsilon^2}}\} }$. 
This is well defined since the original profit in this case is bounded by $w(j,\frac d 2)< b \frac d 2 = \frac 2 {\epsilon^2}\epsilon d$. For later arguments denote the modified work with $w'(j,\frac d 2 ):=w(j,\frac d 2)$ and $w'(j, d):=w'(j,\frac d 2) - p_j'(d)$.
\end{sloppypar}\looseness=-1

For jobs $j$ that are wide in both shelfs, that is when $\gamma'(j,\frac d 2) \ge \gamma'(j, d) \ge b$, we will modify the processing time. In particular we set $t'(j,s) := \frac 1 {1+4\rho} s$ for $s \in \{\frac d 2, d\}$, which results in modified work values $w'(j,s) := t'(j,s) \gamma'(j,s)$. We then define the new profit based on the modified works as: $p_j'(d):= w'(j,\frac d 2) - w'(j,d)$.\looseness=-1

That leaves jobs that are narrow in one shelf and wide in the other. Consider such a job $j$ with  $\gamma'(j,\frac d 2) \ge b > \gamma'(j,d)$. For the narrow version we round again the processing time  $t'(j,\frac d 2) := \frac{1}{1+4\rho} \frac d 2$ and obtain $w'(j,\frac d 2) := t'(j,\frac d 2) \gamma'(j,\frac d 2)$. As for the wide job we round down the work $w(j, \gamma'(j,d))$ to the next multiple of $i\epsilon d$. To be precise we set 
$w'(j,d):= \max\{i\epsilon d \,|\, i\epsilon d \le w(j,\gamma'(j,d)) \text{ and } i\in \mathbb{N}_{\le \frac 4 {\epsilon^2}} \} $. Note that the unmodified work is bounded by $ w(j,\frac d 2)\le w(j,d ) < bd = \frac 4 \epsilon  d  = \frac 4 {\epsilon^2}\epsilon d$. We then obtain the modified profit value $p_j'(d) = w'(j,\frac d 2) - w'(j,d)$. 

With these modified profits and sizes $s'_j(d)= \gamma'(j,d)$ we then solve the resulting problem $KP'(\bigJobs d,m,d,\rho)$ to obtain an optimal item set $J'$.\looseness=-1

\begin{lemma}\label{Lem:}
Let $J'$ be a solution to $KP'(\bigJobs d,m,d,\rho)$ and $d' = (1+4\epsilon)d$, then with unmodified processing times and machine numbers $J'$ is also a solution to $KP(\bigJobs d,m,d')$. Furthermore if there is a schedule with makespan $d$, we have that $W(J',d') \le md' - W(\smallJobs d ,d)$. \looseness=-1
\end{lemma}

\begin{proof}
For the first part we have to show that all jobs in $J'$ fit into the respective knapsack when a processing time of $d'$ or $\frac {d'} 2$ for each shelf is allowed. Consider all jobs $j\in J'$ with $\gamma(j,d)\le b$ and take note that these jobs have the same size in both knapsack instances, since $\gamma(j,d') \le \gamma(j,d)$. For any of the wide jobs $j\in J'$ we have that $t(j,\gamma'(j,d)) \le (1+4\rho)d \le d'$ and therefore $\gamma(j,d') \le \gamma'(j,d)$. We then get $\sum_{j\in J'} \gamma(j,d') \le \sum_{j\in J'} \gamma'(j,d) \le m$ since $J'$ solves the modified knapsack instance which has capacity $m$.\looseness=-1

% additionally if we consider its actual size in the $\rho$-compressible knapsack problem we have with compression \ref{Lem:Comp} that $t(j,\lfloor (1-\rho) \gamma'(j,d) \rfloor) \le (1+4\rho)(1+4\rho)d = (1+\epsilon)^2d \le d'$. Therefore $\gamma(j,d') \le \lfloor (1-\rho) \gamma'(j,d) \rfloor $ and finally $\sum_{j\in J'} \gamma(j,d') \le \sum_{j\in J'} \lfloor (1-\rho) \gamma'(j,d) \rfloor \le m $.

Before we consider the total work of $J'$ we want to make some observations from our rounding: We reduced the number of machines for each job by rounding the sizes. This will only reduce the work of each job due to monotony compared to the original knapsack instance by Mounié et al.. We then only proceed to reduce work further for narrow jobs by at most $\epsilon d$ and reduce the processing time of wide jobs by a factor $\frac 1 {1+4\rho}$.

Note that setting $t'(j,s) = \frac 1 {1+4\rho} s$ for a wide job $j$ and shelf size $s$ is actually reducing processing time and this can be seen through an indirect proof. Assume therefore $t(j,\gamma'(j,s)) < \frac 1 {1+4\rho} s$ and let $s_{k+1} := \gamma'(j,d) $ and let $s_k$ be the next smaller size in $S_\rho$. Reducing the number of machines to $s_k$ is a compression and we then have $t(j,s_k) \le (1+4\rho) t(j,s_{k+1}) <s$. With this $\gamma'(j,s)$ was not chosen minimal. 

In general we have that $w'(j,s) \le w(j,s)$ and want to continue to give an upper bound on $w(j,s)$. Note that we may assume that processing times do not increase with increasing numbers of machines. Otherwise we could simply omit numbers of machines that increase processing times and always schedule on the smaller number. With this we get that $w(j,s) \le (1+\rho)\gamma'(j,s)t(j,\gamma'(j,s))$.

Note that for jobs $j$ in shelf 2 we only decrease the processing time if they are wide and therefore we get: 
\begin{equation*}
w(j,\frac d 2) \le (1+\rho)\gamma'(j,\frac d 2) (1+4\rho)t'(j,\gamma'(j,\frac d 2))=(1+\rho)(1+4\rho)w'(j,\frac d 2). 
\end{equation*}
For wide jobs $j$ in shelf 1 we do the same. However for narrow jobs of this shelf we reduce the work further by $\epsilon d$. Doing the same estimation for $w(j,d)$ that we did for $w(j,\frac d 2)$ and adding this additional increase, we can conclude that:
${w(j,d) \le (1+\rho)(1+4\rho)w'(j,d) +\epsilon d}$.

%For any job we have that $w'(j,\frac d 2) \le w(j,\frac d 2) \le (1+\rho)(1+4\rho)w'(j,\frac d 2)$, which is a result of rounding processor sizes and processing times. For the shelf of size $d$ we can give a similar estimation, but we need to account for the additional error by rounding the profit or work directly. Summarizing both cases however we can conclude that 
%$w'(j,d) \le w(j,d) \le (1+\rho)(1+4\rho)w'(j,d) +\epsilon d$.\looseness=-1

For the second part of the statement we get through lemma \ref{Lem:schedtoks} that there is an optimal solution $J^*$ to $KP(\bigJobs d,m,d)$ with $W(J^*,d) \le md - W(\smallJobs d,d)$. Further $J^*$ is also a feasible solution for the modified knapsack problem, since our modifications only reduce item sizes. Our modified knapsack instance, similar to the original one, will maximize knapsack profits, which in turn then minimizes total work of a two-shelf schedule with modified work values. Since $J'$ is an optimal solution of the modified instance, we have that the total modified work of $J'$ is larger than the modified work of $J^*$. To be precise we have:
\begin{equation*}
{ \sum_{j \in J'}{w'(j,d)} + \sum_{j\in \bigJobs{d} \backslash J^*} {w'(j, \frac d 2 )}) \le
 \sum_{j \in J^*}{w'(j,d)} + \sum_{j\in \bigJobs{d} \backslash J^*} {w'(j, \frac d 2 )}) }.
 \end{equation*}

We now can conclude that the total work of the two-shelf schedule implied by $J'$ is bound:\looseness=-1
\begin{align*}
W(J',d) &= \sum_{j \in J'}{w(j,d)} + \sum_{j\in \bigJobs{d} \backslash J'} {w(j, \frac d 2 )}\\
%&\le \sum_{j \in J'}{w(j,d)} + \sum_{j\in \bigJobs{d} \backslash J'} {w(j, \frac d 2 )}\\
&\le \sum_{j \in J'}{((1+\rho)(1+4\rho)w'(j,d) +\epsilon d)} + \sum_{j\in \bigJobs{d} \backslash J'} {(1+\rho)(1+4\rho)w'(j, \frac d 2 )}\\
&\le |J'| \epsilon d +  (1+\rho)(1+4\rho)(\sum_{j \in J'}{w'(j,d)} + \sum_{j\in \bigJobs{d} \backslash J'} {w'(j, \frac d 2 )})\\
&\le |J'| \epsilon d +  (1+\rho)(1+4\rho)(\sum_{j \in J^*}{w'(j,d)} + \sum_{j\in \bigJobs{d} \backslash J^*} {w'(j, \frac d 2 )})\\
&\le |J'| \epsilon d +  (1+\rho)(1+4\rho)(\sum_{j \in J^*}{w(j,d)} + \sum_{j\in \bigJobs{d} \backslash J^*} {w(j, \frac d 2 )})\\
&\le m \epsilon d +  (1+\rho)(1+4\rho)(md -W(\smallJobs d ,d)) \\
&\le (1+4\epsilon)md -W(\smallJobs d ,d) = md' -W(\smallJobs d ,d) 
\end{align*}\looseness=-1
Lastly due to monotony of work we have also that $W(J',d') \le W(J',d)$, which concludes the proof.\looseness=-1

\end{proof}

\subsection{Solving the knapsack problems}
%Knapsack model:
%\begin{itemize}
%\item Solve Knapsack problem with compressible items. Knapsack problem given through capacity $t=m$ and for each job $j$ we have an item with size $s_j= \gamma'(j,d)$ and profit $p_j = W(j,\gamma'(j,d/2) - W(j,\gamma'(j,d))$	
	
%\end{itemize}

As we already mentioned we intend to use an algorithm from Axiotis and Tzamos \cite{axiotis19}. Their algorithm works in two main steps. In the first step the items of the knapsack instance are partitioned into sets containing items of equal size. The knapsack problem is then solved for each item set separately and for every item size $s$ with item set $I_s = \{i \in I \,|\, s_i = s\}$ a solution array $R_s$ is generated where $R_s[t']$ denotes the maximum profit achievable for a knapsack of size $t' \le t$ using only items with size $s$. Note that by the nature of this problem $R_s[t']$ will always be given by the sum of profits of the $\lfloor \frac{t'}{s} \rfloor$ items with the highest profit in $I_s$.\looseness=-1

These solution arrays $R_s $ have a special structure as $R_s[k\cdot s] = R_s[k\cdot s +s']$ for all $s'< s$ and $k\in \mathbb{N}$. Further considering the unique entries we have that $R_s[(k+1)\cdot s]-R_s[k\cdot s] \ge R_s[(k+2)\cdot s]-R_s[(k+1)\cdot s]$ for each $k$, since the profit of the items added decreases. This structure is also called \textit{$s$-step concave} as the unique entries build a concave sequence. 
In the second step of their algorithm they combine the solution arrays in sequential order via convolution to generate a final solution array $R = R_1 \oplus R_2 \oplus \cdots \oplus R_{[s_{max}]}$.\looseness=-1

A very important result from Axiotis and Tzamos is that if these convolutions are done in sequential order, then one sequence will always be $s$-concave for some respective $s$. They proved in their paper that convolution with one $s$-step-concave sequence can be done in linear time, opposed to the best known quadratic time.\looseness=-1

%\begin{lemma}\label{lem:concaveConv}
%Given two sequences or arrays $A,B$ of length $n$ such that $B$ is $s$-concave for some $s\in [n]$ then we can compute the convolution $A\oplus B$ in time $O(n)$
%\end{lemma}

\begin{lemma}[\cite{axiotis19}]\label{lem:concaveConv}
Given any sequence $A$ and $R_h$ for some $h\in \mathbb{N}$, each with $t$ entrys, we can compute the convolution $A\oplus R_h$ in time $O(t)$. \looseness=-1
\end{lemma}

In our setting the knapsack capacity is given by $t=m$. Thanks to our rounding we only have $|S_\rho|$ different item sizes, which defines the number of convolutions we have to calculate. We however must also compute the initial solutions that consist of the highest profit items for each size. Thanks to rounding item profits we can also sort these efficiently to generate the initial solutions arrays $R_h$.\looseness=-1

% Axiotis and his co-authors don't really specify this as they assumed that items are already sorted. Polak et al. \cite{polak21} proposed an approach using order statistics to find the highest profitable items and then sorting this reduced item set in a total time of $O(n+t\log^2(t))$. Since we have a very specially structured instance, we can do the sorting faster:

%Based on our rounding techniques we can bound the number of item sizes and the number of relevant items yielding a modified running time which we will specify for each of the steps of the algorithm. 
%The main idea of Axiotis et al. is to partition the items into sets holding items of the same size. Given such an item set $I_h$ only containing items $i\in I_h$ with $s_i = h$ we can then solve the knapsack problem only limited to these items rather simple, by taking the $\lfloor \frac{t}{h}\rfloor$ most profitable items. Note that we can also solve the problem for all potential smaller capacities that way. The goal is compute arrays for every item size $h$ $R_h$ where $R_h[t']$ for any $t' \le t$ denotes the maximum possible profit of a knapsack with size $\le t'$ only using items from $I_h$.
%Polak et al. gave a worst-case analysis on the running time of computing these arrays, but due to our rounding procedure we can give a better analysis.

\begin{lemma}
Given a modified knapsack instance $KP'(\bigJobs d,m,d,\rho)$, we can compute for all $t \le t$ the entry $R_h[t']$ in time 
$ O(n+m(\frac 1\epsilon + \frac{\log (\epsilon m)}{\epsilon}))$.\looseness=-1
\end{lemma}

\begin{proof}
Our goal is to sort items by profits and subsequently add up the highest profits to fill the arrays $R_h$. We will sort items based on how they were rounded:\looseness=-1
%The main problem is that we need to sort the items by their profits. When we have done that we can simply subsequently take the highest profit item and fill the arrays $R_h$. 
%For sorting the items we will make use of the rounded structure and reduced number of item profits. 

Consider jobs $j$ with $\gamma'(j,s)< b$ for both $s\in \{\frac d 2, d\}$ and denote the number of these jobs with $n_1$. By scaling their profits with $\frac 1 \epsilon \frac 1 d $ we obtain profits of the form $\tilde{p}_j(d)= i$ for some  $i\in \mathbb{N}_{\le  \frac 2 {\epsilon^2}}$. We can sort profits using radix sort in time $O(n_1+\frac 1 \epsilon)$ where we encode them using $O(1)$ digits ranging from $0$ to $\frac 1 \epsilon$.\looseness=-1

Consider now the $n_2$ jobs $j$ with $\gamma'(j,\frac d 2) \ge \gamma'(j, d) \ge b$. If we scale the profit of these items with $\frac{1+4\rho}{d}$ then we have that $\tilde{p}_j(d) =\frac 1 2\gamma'(j,\frac d 2) - \gamma'(j,d) $. These items can be sorted by profit using bucket sort in $O(n_2+m)$. \looseness=-1

For the remaining $n_3$ of the jobs $j$ with $\gamma'(j,\frac d 2) \ge b > \gamma'(j,d)$ we have to consider the modified profits $p_j'(d) := \frac d {2(1+4\rho)}\gamma'(j,\frac d 2) - i\epsilon d $ for some $i\in \mathbb{N}$. We scale these profits with $\frac{2(1+\epsilon)}{d \epsilon^2}$ to obtain $\tilde{p}_j(d) = \gamma'(j,\frac d 2)\frac 1 {\epsilon^2} - \frac{2id}{\epsilon}  \le \frac{m}{\epsilon^2}$. These items can be sorted with radix sort in time $O(n_3+\frac m \epsilon)$ by encoding profits with two digits ranging from $0$ to $\frac m \epsilon$.\looseness=-1

Putting these three steps together takes time $O(n_1+n_2+n_3+\frac 1 \epsilon + m +\frac m \epsilon) = O(n+\frac m \epsilon)$. We can additionally merge the three sorted lists via merge sort in $O(n)$ and iterate through all items to fill the actual solution arrays. The number of total entries we have to fill in is at most $m(\frac 4\epsilon + \frac{\log (\epsilon m)}{\epsilon})$ since we have $m$ entries in each array, and one array for every item size.\looseness=-1
\end{proof}

Technically we only need the unique entries of these solution arrays to apply the algorithm \cite{polak21}. These could effectively be calculated in time $O(n+\frac m \epsilon)$ but combining all arrays will dominate the running time regardless.
%in , which could effectively be computed in 
%We note that filling these actual arrays with entries is not necessary. It suffices for the next part to compute the unique profit values that each array contains and it is possible to use this information alone to do the following steps \cite{polak21}. This would reduce the time down to $O(n+\frac m \epsilon)$ but the next part will dominate the running time. We can now apply lemma \ref{lem:concaveConv} and sequentially build the convolution of all the solution arrays.\looseness=-1

%Since we have computed all the solutions for the partitioned knapsack instance, we only need to combine these solutions by computing the $(\max ,+)$-convolution of these arrays.  Axiotis and Tzamos. showed that the respective convolution can be computed in linear time in the number of array entries, due to the simple structure of the solution arrays $R_h$. Note that the Array $R_h$ is $h$-concave and therefore the convolution problem can be reduced to finding the row maxima of a inverse monge matrix. Finding all required entries can be done in time $O(t)$ resulting in the following lemma from Axiotis et al. :

%By simply applying this procedure to all $R_h$ we can combine all the knapsack solutions and gain a solution for the total problem with all items.

\begin{corollary}\label{coro:conv}
We can compute $R_1 \oplus R_2 \oplus \cdots \oplus R_{|S_\rho|}$ in time $O(m(|S_\rho|))$.\looseness=-1
\end{corollary}

With this knapsack solution we can construct a schedule using corollary \ref{coro:final}. We note that this final construction using the procedure from Mounié et al.~\cite{mounieTrystram} can be implemented in time $O(n)$ by using rounded processing times\cite{jansenLandLinTime}. 
%The solution of the knapsack problem is the $m$-th entry of the resulting array and with it we can construct a two-shelf schedule, which can be modified to a feasible solution of moldable machine scheduling. 
\looseness=-1

%With the knapsack problem being solved, our algorithm only needs to translate the knapsack solution into a two-shelf schedule and construct a feasible solution based on the procedure from Mounie et al  \cite{mounieTrystram}. We can now put together these pieces and give an analysis of the running time of our algorithm 

%\begin{theorem}\label{Theo:ThreeHalf}
%For machine scheduling with moldable jobs there exists an algorithm that for any $\epsilon >0$ yields a $\frac 3 2 +\epsilon$ approximation in time $O(n\log (\frac 1 \epsilon) \log(m+ \log (\frac 1 %\epsilon )) + n \frac 1 \epsilon\log (\epsilon m))$.%
%\end{theorem}

\begin{proof}[of Theorem \ref{Theo:ThreeHalf}]

We apply the dual approximation framework, which means we compute an upper bound for $d$ in time $O(n \log^2(\frac 1 \epsilon + \frac {\log (\epsilon m)} \epsilon))$.
We end up with $\log(\frac 1 \epsilon)$ candidates for $d$ and construct knapsack instances for all of them.

To do so we need to identify their machine count among compressed sizes. This can be done in $O(n\log(\frac 1 \epsilon + \frac{\log (\epsilon m)} \epsilon))$ via binary search. All further modifications to knapsack items can be done in $O(n)$. In total for all candidates these steps take time $O(n \log^2(\frac 1 \epsilon + \frac {\log (\epsilon m)} \epsilon))$.

Solving the resulting knapsack problem for one candidate can be done in time $O(m(\frac 1\epsilon + \frac{\log (\epsilon m)}{\epsilon})) \subseteq O(m\frac 1\epsilon \log (\epsilon m))$. By applying this to all candidates and since $m\le 16n$ we get a final running time of 
$O(n \log^2(\frac 1 \epsilon + \frac {\log (\epsilon m)} \epsilon) +  \frac{n}{\epsilon} \log(\frac 1 \epsilon) {\log (\epsilon m)}) $.

%Considering one candidate $d$ we have to find rounded processing times. This can be done in $O(n\log(\frac 1 \epsilon + \frac{\log (\epsilon m)} \epsilon))$ via binary search. We then construct the knapsack instance in $O(n)$

%and the initial solution arrays in $O(n + m(\frac 1\epsilon + \frac{\log (\epsilon m)}{\epsilon}))$. Since $m\le 16n$ all these steps for all candidates take time  $O(n \log^2(\frac 1 \epsilon + \frac {\log (\epsilon m)} \epsilon))$.

%We then also compute the convolutions for all candidates. For one candidate this is done in $O(m(\frac 1\epsilon + \frac{\log (\epsilon m)}{\epsilon})) \subseteq O(m\frac 1\epsilon \log (\epsilon m))$. By applying this to all candidates we get a final running time of 
%$O(n \log^2(\frac 1 \epsilon + \frac {\log (\epsilon m)} \epsilon) +  \frac{n}{\epsilon} \log(\frac 1 \epsilon) {\log (\epsilon m)}) $.

\looseness=-1

\end{proof}

\section{Implementation}

We implemented all algorithms introduced and used in this paper, along with a version of the algorithm introduced by Jansen and Land \cite{jansenLandLinTime}. We note that we did not implement the final version of their algorithm to solve Knapsack with compressible items, as it was very intricate and complicated. Instead our implementation computes their modified knapsack instance and solves it via their proposed dynamic programming approach. 

 The implementations and experiments were conducted on a Raspberry Pi 4 Model B and we limited the experiment to one CPU-core as we did not use any mean of parallelization. We uploaded a version of our implementation to GitHub (https://github.com/Felioh/MoldableJobScheduling). In the following we mainly tested for the part where $m\le 16n$ as we deem this the more relevant comparison between the two results. 

\subsection{Computational results}

As for test instances we generated sets of randomized instances for moldable job scheduling. Machine numbers mainly range from 30 to 100 and jobs from 10 to 120. We tested on these instances for $\epsilon = \frac 1 {10}$ and the results can be seen in the figures in the appendix. Figures \ref{fig:runtimeJobs} and \ref{fig:runtimeMach} show the difference of average runtime between our algorithm and the one by Jansen and Land. Note that the runtime of our algorithm is subtracted from the runtime of their algorithm. Hence we can see that our algorithm does slightly better for the analyzed number of jobs and machines and that our algorithm seems to scale better with growing numbers of machines and jobs.

In figures \ref{fig:msmach1} through \ref{fig:msmach3} we compare the average makespans of both algorithms to compare solution quality. In most cases that solution quality is generally quite similar but in some cases slightly better for our algorithm. We believe that our algorithm does better in regards to solution quality due to our rounding. For one our rounding of machine numbers to values in $S_\rho$ is in its core a compression but does not fully utilize the potential introduced in lemma \ref{Lem:Comp}. Since we do not reduce the machine counts by the maximal possible amount, our effective error is smaller. In a similar manner are the additional modifications of knapsack items mainly catered to achieving a simple structure that also keeps the additional error small.

\section{Conclusion}

In this paper we presented our new $\frac 3 2 + \epsilon$-approximation, that results from the combination of different techniques from moldable scheduling, knapsack and convolution. Our algorithm gives a theoretical improvement in terms of the known upper bound for this problem, but also proves to be faster in practice as shown by our experiments. An interesting takeaway from our result is that it is sufficient to reduce moldable scheduling to only a certain set of machine counts thanks to compression. In fact it is not necessary to regard all possible allotments, when one wants to find an approximate solution.

%The most important conclusion we can draw is that for moldable scheduling it is possible to reduce the number of relevant processor numbers, if we allow a small multiplicative error. Using this we were able to improve the runtime of the FPTAS from Jansen and Land for many machines. We further simplified the knapsack problem required to find a two shelf schedule and gave a simple rounding scheme for processor numbers, work and processing times. This allowed us to sort the resulting items by their profits fast and with that we were able to solve the knapsack problem in an efficient way using $(\max,+)$-convolution, since we don't have many item sizes.\looseness=-1

%
% ---- Bibliography ----
%
% BibTeX users should specify bibliography style 'splncs04'.
% References will then be sorted and formatted in the correct style.
%
\bibliographystyle{splncs04}
\bibliography{lib}

\begin{thebibliography}{10}
\providecommand{\url}[1]{\texttt{#1}}
\providecommand{\urlprefix}{URL }
\providecommand{\doi}[1]{https://doi.org/#1}

\bibitem{axiotis19}
Axiotis, K., Tzamos, C.: {Capacitated Dynamic Programming: Faster Knapsack and
  Graph Algorithms}. In: 46th International Colloquium on Automata, Languages,
  and Programming (ICALP 2019). (LIPIcs), vol.~132, pp. 19:1--19:13. Dagstuhl,
  Germany (2019)

\bibitem{bateni}
Bateni, M., Hajiaghayi, M., Seddighin, S., Stein, C.: Fast algorithms for
  knapsack via convolution and prediction. In: Proc. of the 50th Annual ACM
  SIGACT Symposium on Theory of Computing. p. 1269–1282. STOC 2018, New York,
  NY, USA (2018)

\bibitem{Belkhale}
Belkhale, K.P., Banerjee, P.: An approximate algorithm for the partitionable
  independent task scheduling problem. In: International Conference on Parallel
  Processing (ICPP). pp. 72--75 (1990)

\bibitem{bellman}
Bellman, R.: Dynamic programming. In: Princeton University Press (1957)

\bibitem{cygan}
Cygan, M., Mucha, M., Wegrzycki, K., Wlodarczyk, M.: { On Problems Equivalent
  to (min,+)-Convolution}. In: 44th International Colloquium on Automata,
  Languages, and Programming (ICALP 2017). (LIPIcs), vol.~80, pp. 22:1--22:15
  (2017)

\bibitem{32Complexity}
Drozdowski, M.: On the complexity of multiprocessor task scheduling. Bulletin
  of The Polish Academy of Sciences-technical Sciences  \textbf{43},  381--392
  (1995)

\bibitem{mallComplexity}
Du, J., Leung, J.Y.T.: Complexity of scheduling parallel task systems. SIAM
  Journal on Discrete Mathematics  \textbf{2}(4),  473--487 (1989)

\bibitem{eisenbrand}
Eisenbrand, F., Weismantel, R.: Proximity results and faster algorithms for
  integer programming using the steinitz lemma. ACM Trans. Algorithms
  \textbf{16}(1) (nov 2019)

\bibitem{gareyGraham}
Garey, M.R., Graham, R.L.: Bounds for multiprocessor scheduling with resource
  constraints. SIAM J. Comput.  \textbf{4},  187--200 (1975)

\bibitem{jansenLandLinTime}
Jansen, K., Land, F.: Scheduling monotone moldable jobs in linear time. In:
  2018 IEEE International Parallel and Distributed Processing Symposium
  (IPDPS). pp. 172--181. IEEE Computer Society, Los Alamitos, CA, USA (may
  2018)

\bibitem{jansenLandNP13}
Jansen, K., Land, F., Land, K.: Bounding the Running Time of Algorithms for
  Scheduling and Packing Problems, Bericht des Instituts für Informatik,
  vol.~1302 (2013)

\bibitem{ludwigTiwari}
Ludwig, W., Tiwari, P.: Scheduling malleable and nonmalleable parallel tasks.
  In: Proc. of the Fifth Annual ACM-SIAM Symposium on Discrete Algorithms. p.
  167–176. SODA '94, Society for Industrial and Applied Mathematics, USA
  (1994)

\bibitem{mounieTrystram}
Mounié, G., Rapine, C., Trystram, D.: A 3/2-dual approximation for scheduling
  independant monotonic malleable tasks. SIAM J. Comput.  \textbf{37},
  401--412 (01 2007)

\bibitem{polak21}
Polak, A., Rohwedder, L., W\k{e}grzycki, K.: {Knapsack and Subset Sum with
  Small Items}. In: 48th International Colloquium on Automata, Languages, and
  Programming (ICALP 2021). (LIPIcs), vol.~198, pp. 106:1--106:19. Dagstuhl,
  Germany (2021)

\bibitem{Turek}
Turek, J., Wolf, J.L., Yu, P.S.: Approximate algorithms scheduling
  parallelizable tasks. In: Proc. of the Fourth Annual ACM Symposium on
  Parallel Algorithms and Architectures. p. 323–332. SPAA '92, Association
  for Computing Machinery, New York, NY, USA (1992)

\bibitem{WuZhangChen2022}
Wu, F., Zhang, X., Chen, B.: An improved approximation algorithm for scheduling
  monotonic moldable tasks. European Journal of Operational Research
  \textbf{306}(2),  567--578 (2023).
  \doi{https://doi.org/10.1016/j.ejor.2022.08.034},
  \url{https://www.sciencedirect.com/science/article/pii/S0377221722006762}

\end{thebibliography}

\newpage
\begin{appendix}

\section{Computational Results (Graphs and Diagrams)}

\begin{figure}
\includegraphics[width=\textwidth]{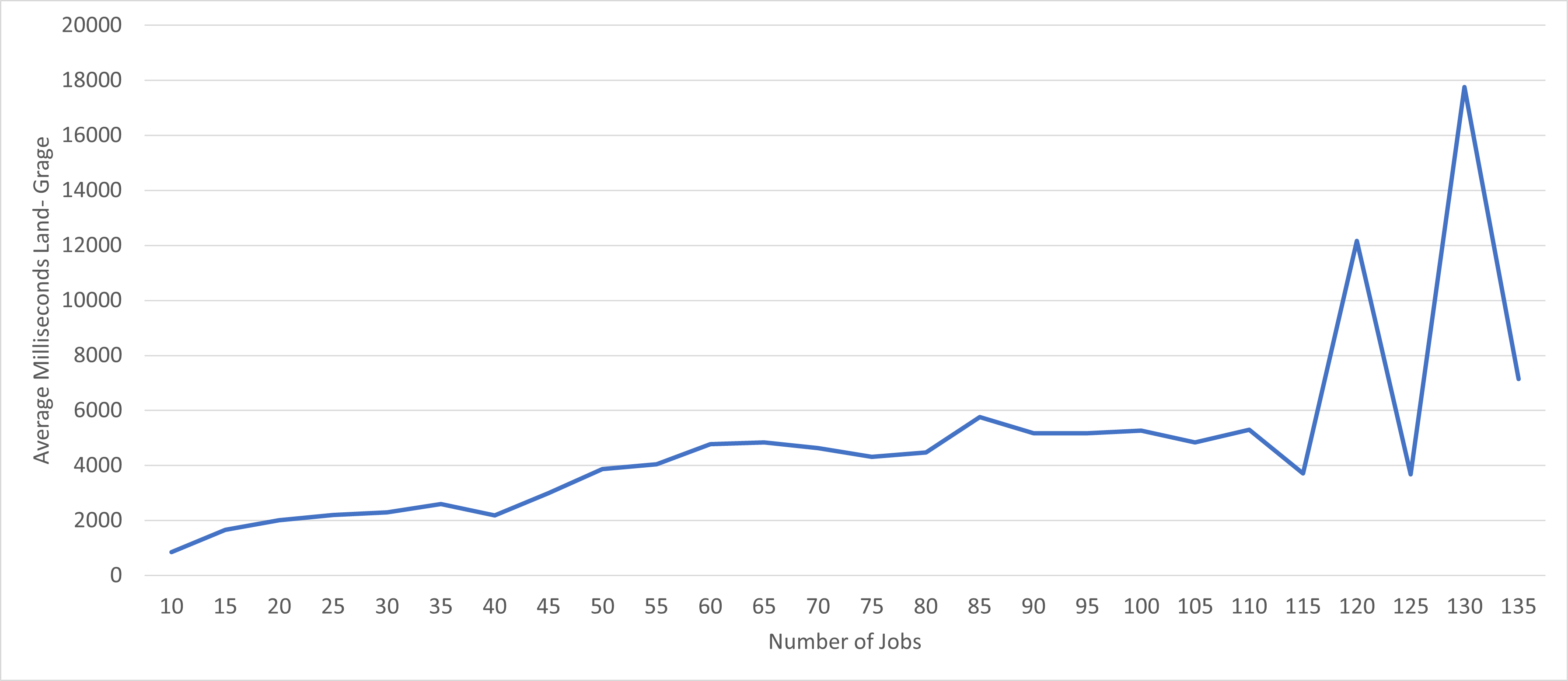}
\caption{Average runtime difference in relation to job numbers. } \label{fig:runtimeJobs}
\end{figure}

\begin{figure}
\includegraphics[width=\textwidth]{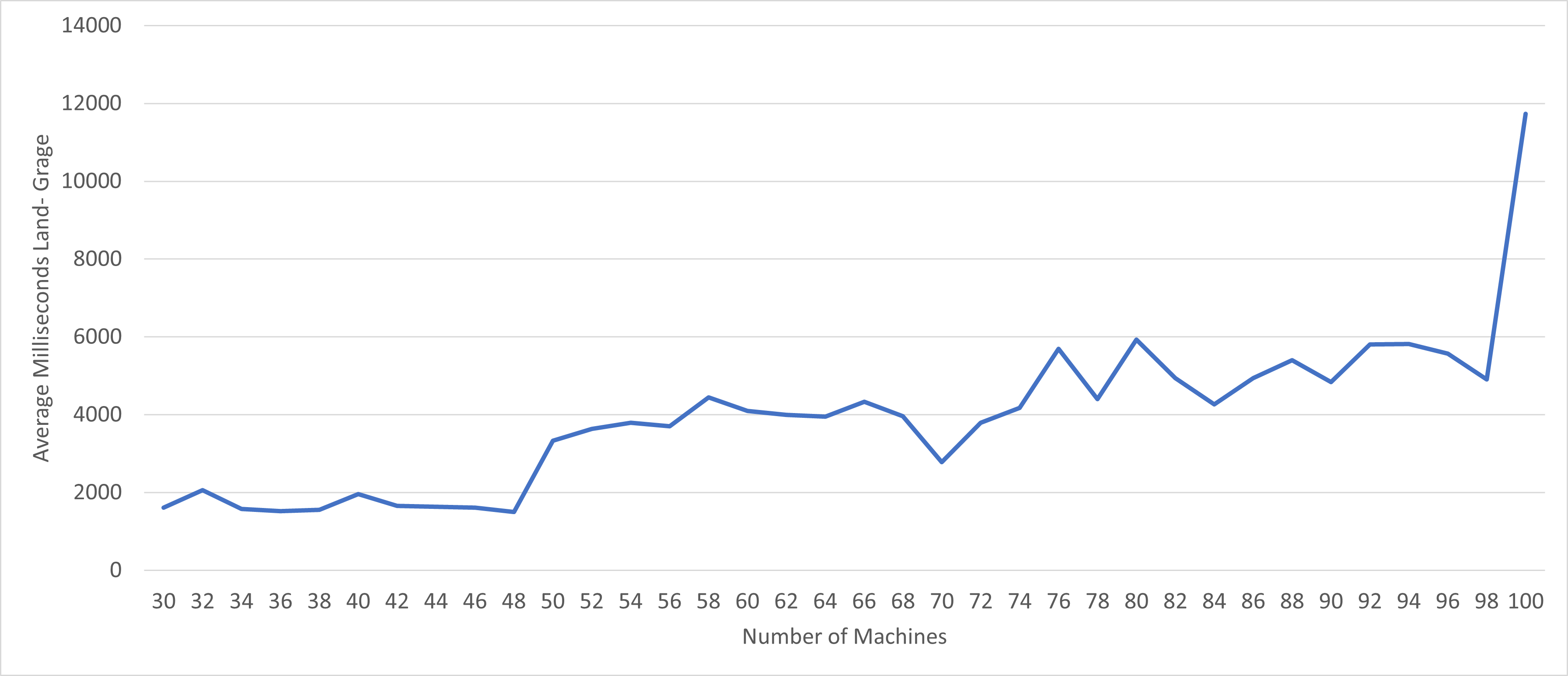}
\caption{Average runtime difference in relation to machine numbers.} \label{fig:runtimeMach}
\end{figure}

\begin{figure}
\includegraphics[width=\textwidth]{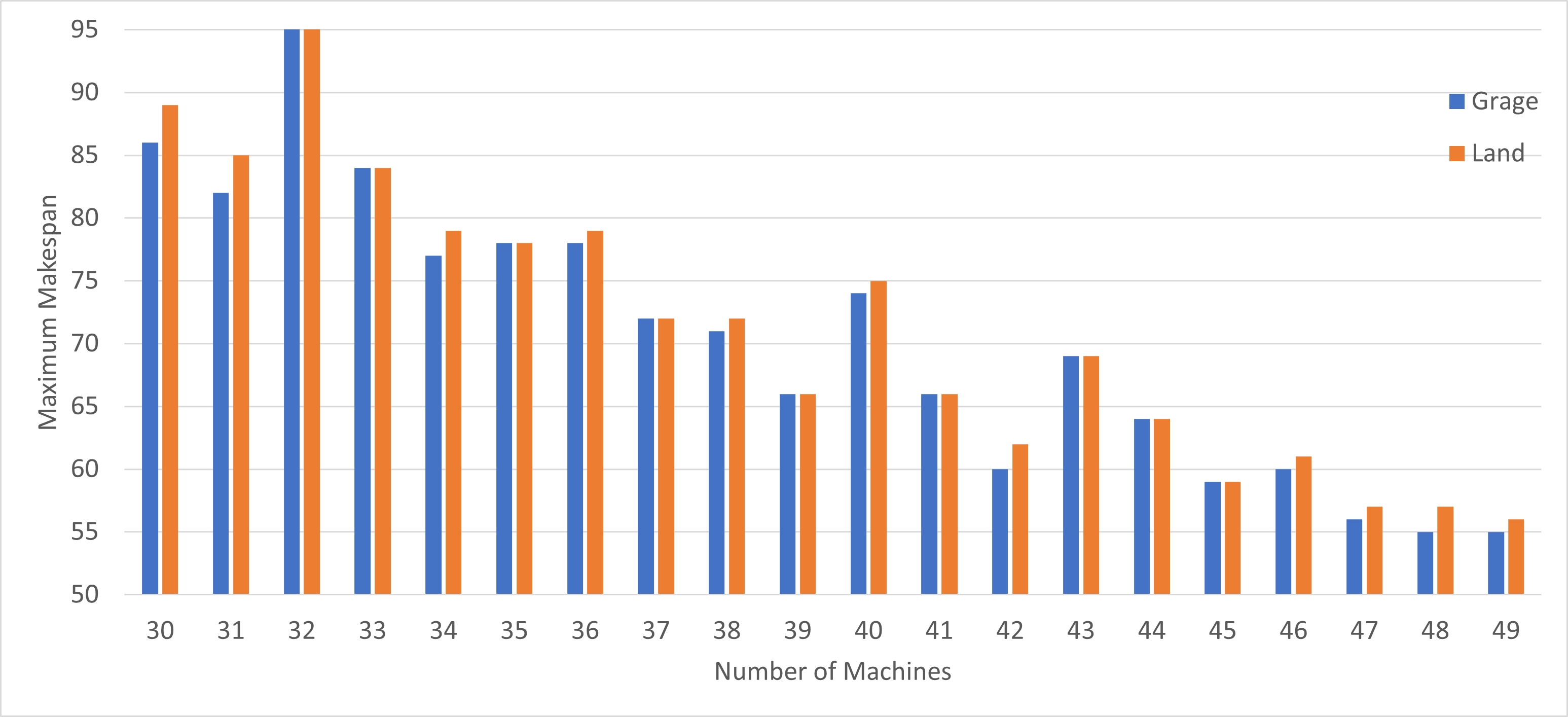}
\caption{Average makespan comparison limited to instances with same machines.} \label{fig:msmach1}
\end{figure}

\begin{figure}
\includegraphics[width=\textwidth]{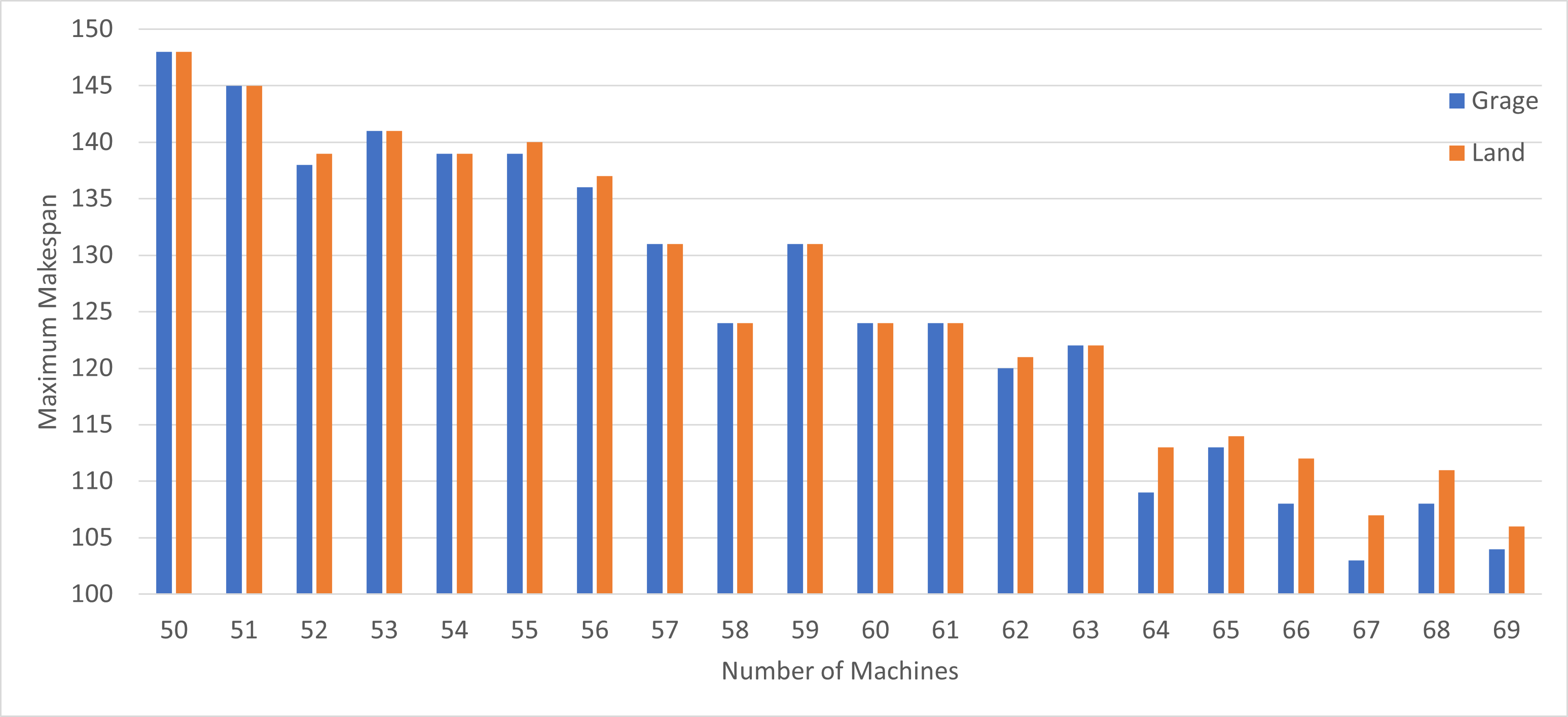}
\caption{Average makespan comparison limited to instances with same machines.} \label{fig:msmach2}
\end{figure}

\begin{figure}
\includegraphics[width=\textwidth]{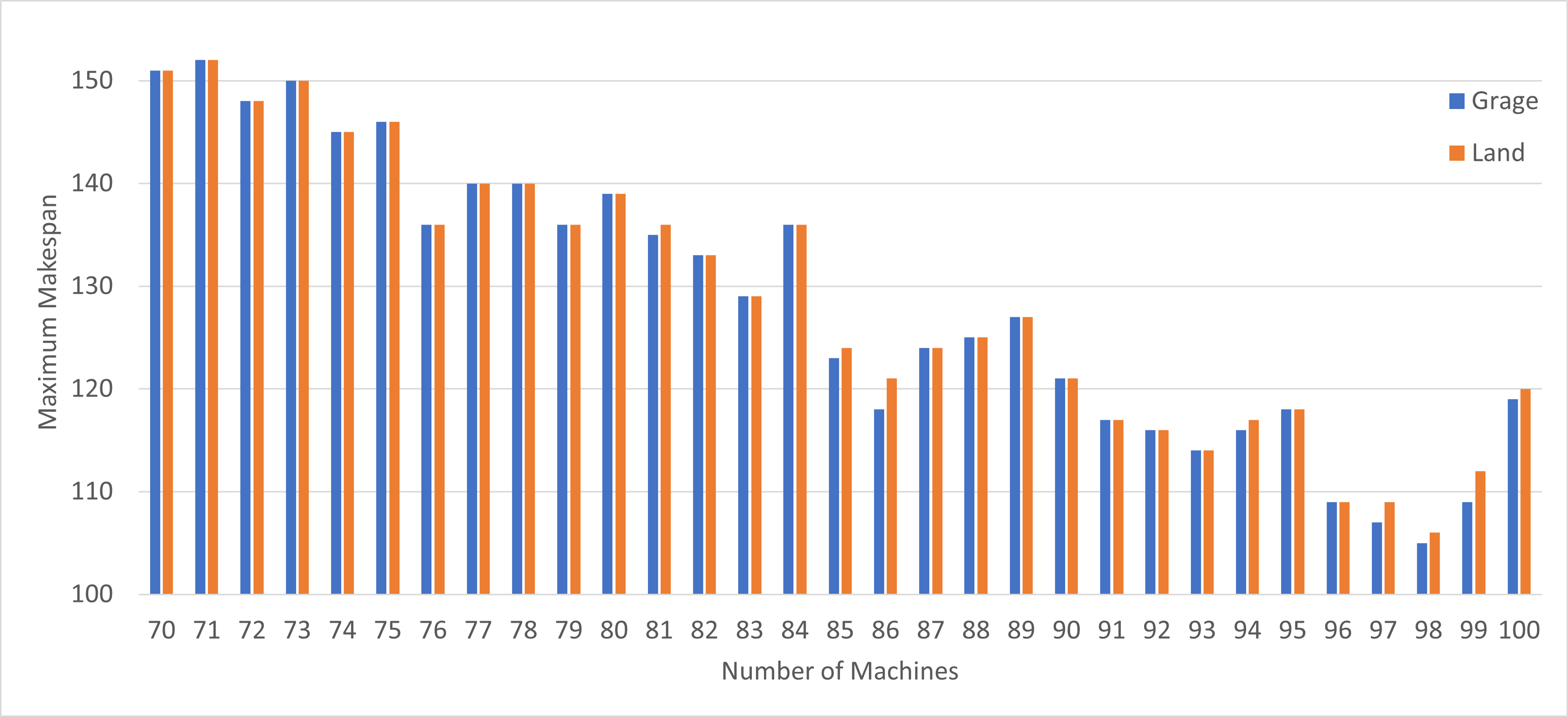}
\caption{Average makespan comparison limited to instances with same machines.} \label{fig:msmach3}
\end{figure}

\end{appendix}
\end{document}